\documentclass{amsart}

\usepackage[utf8]{inputenc}
\usepackage[american]{babel}
\usepackage{amsmath,amssymb,amsthm}
\usepackage{mathrsfs} 
\usepackage{bm}
\usepackage{microtype}
\usepackage{csquotes}
\usepackage[style=alphabetic,sorting=anyt,backend=bibtex,doi=false,url=false,maxbibnames=10,isbn=false]{biblatex}
\usepackage[hidelinks]{hyperref}

\newcommand{\A}{\mathcal{A}}
\newcommand{\IC}{\mathbb{C}}
\newcommand{\D}{\mathcal{D}}
\renewcommand{\H}{\mathfrak{H}}
\newcommand{\Hh}{\mathcal{H}^{(h)}}
\newcommand{\I}{\mathcal{I}}
\newcommand{\J}{\mathcal{J}}
\renewcommand{\L}{\mathscr{L}}
\newcommand{\IR}{\mathbb{R}}
\renewcommand{\S}{\mathfrak{S}}
\newcommand{\bV}{\mathbf{V}}
\newcommand{\W}{\mathcal{W}}
\newcommand{\bW}{\mathbf{W}}

\renewcommand{\epsilon}{\varepsilon}
\renewcommand{\phi}{\varphi}

\newcommand{\lra}{\longrightarrow}

\newcommand{\CE}{\mathsf{CE}}
\renewcommand{\div}{\operatorname{div}}

\newcommand{\HJB}{\mathsf{HJB}}

\newcommand{\abs}[1]{\lvert#1\rvert}
\newcommand{\norm}[1]{\lVert#1\rVert}

\theoremstyle{remark}
\newtheorem{definition}{Definition}[section]

\theoremstyle{plain}
\newtheorem{theorem}[definition]{Theorem}
\newtheorem*{theorem*}{Theorem}
\newtheorem{lemma}[definition]{Lemma}
\newtheorem{proposition}[definition]{Proposition}

\begin{document}

\title[A dual formula for the quantum transport distance]{A dual formula for the noncommutative transport distance}
\author{Melchior Wirth}
\address{IST Austria\\Am Campus 1\\3400 Klosterneuburg\\Austria}
\email{melchior.wirth@ist.ac.at}

\begin{abstract}
In this article we study the noncommutative transport distance introduced by Carlen and Maas and its entropic regularization defined by Becker and Li. We prove a duality formula that can be understood as a quantum version of the dual Benamou--Brenier formulation of the Wasserstein distance in terms of subsolutions of the Hamilton--Jacobi--Bellmann equation.
\end{abstract}

\maketitle

\section*{Introduction}

The theory of optimal transport \cite{Vil03,Vil09} has experienced rapid growth in recent years with applications in diverse fields across pure and applied mathematics. Along with this growth came a lot of interest in extending the methods of optimal transport beyond the scope of its original formulation as an optimization problem for the transport cost between two probability measures.

One such extension deals with ``quantum spaces'', where the probability measures are replaced by density matrices or density operators. Most of the work dealing with quantum optimal transport in this sense can be grouped into one of the following two categories. The first approach (see e.g. \cite{CM14a,CM17a,MM17,CGT18,DR19,BV20,CM20}) relies on a noncommutative analog of the Benamou--Brenier formulation \cite{BB00} of the Wasserstein distance for probability measures on Euclidean space
\begin{equation*}
W_2^2(\mu,\nu)=\inf\left\lbrace\int_0^1\int_{\mathbb{R}^n}\abs{v_t}^2\,d\rho_t\,dt: \rho_0=\mu,\rho_1=\nu,\dot\rho_t+\nabla\cdot(\rho_t v_t)=0\right\rbrace.
\end{equation*}
This approach has proven fruitful in applications to noncommutative functional inequalities, similar in spirit to the heuristics known as Otto calculus \cite{CM17a,CM20,DR20,WZ20}.

The second approach (see e.g. \cite{NGT15,GMP16,PCVS19,Duv20,PMTL20,PT21}) seeks to find a suitable noncommutative analog of the Monge--Kantorovich formulation \cite{Kan42} of the Wasserstein distance via couplings (or transport plans):
\begin{equation*}
W_p^p(\mu,\nu)=\inf\left\lbrace\int_{X\times X}d^p(x,y)\,d\pi(x,y):(\mathrm{pr}_1)_\#\pi=\mu,(\mathrm{pr}_2)_\#\pi=\nu\right\rbrace.
\end{equation*}
This approach also allows to consider a quantum version of Monge--Kantorovich problem for arbitrary cost functions. So far, possible connections between these two approaches in the quantum world stay elusive.

The focus of this article lies on the noncommutative transport distance $\W$ introduced in the first approach. More precisely, we prove a dual formula that is a noncommutative analog of the expression of the classical $L^2$-Wasserstein distance in terms of subsolutions of the Hamilton--Jacobi equation \cite{OV00,BGL01}
\begin{equation*}
W_2^2(\mu,\nu)=\frac 1 2 \inf\left\lbrace\int_{\mathbb{R}^n} u_1\,d\mu-\int_{\mathbb{R}^n}u_0\,d\nu: \dot u_t+\frac 1 2\abs{\nabla u_t}^2\leq 0\right\rbrace.
\end{equation*}
This result yields a noncommutative version of the dual formula obtained independently by Erbar, Maas and the author \cite{EMW19} and Gangbo, Li and Mou \cite{GLM19} for the Wasserstein-like transport distance on graphs. In fact, we prove a dual formula that is not only valid for the metric $\W$, but also for the entropic regularization recently introduced by Becker--Li \cite{BL21}.

With the notation introduced in the next section, the main result of this article reads as follows.
\begin{theorem*}
Let $\sigma\in M_n(\IC)$ be an invertible density matrix and $(P_t)$ an ergodic quantum Markov semigroup on $M_n(\IC)$ that satisfies the $\sigma$-DBC. The entropic regularization $\W_\epsilon$ of noncommutative transport distance induced by $(P_t)$ satisfies the following dual formula:
\begin{equation*}
\frac 1 2\W_\epsilon^2(\rho_0,\rho_1)=\sup\{\tau(A(1)\rho_1-A(0)\rho_0)\mid A\in \HJB^1_{\epsilon}\}.
\end{equation*}
\end{theorem*}
Here $\HJB^1_{\epsilon}$ stands for the set of all Hamilton--Jacobi--Bellmann subsolutions, a suitable noncommutative variant of solutions of the differential inequality
\begin{equation*}
\dot u(t)+\frac 1 2\abs{\nabla u(t)}^2-\epsilon \Delta u(t)\leq 0.
\end{equation*}

Other metrics similar to $\W$ also occur in the literature, most notably the one called the ``anticommutator case'' in \cite{CGT18,CGGT20,BL21}. In \cite{Wir18,CM20}, a class of such metrics was studied in a systematic way, and our main theorem applies in fact to this wider class of metrics. For the anticommutator case, this duality formula was obtained before in \cite{CGGT20}.

There are still some very natural questions left open. For one, we do not discuss the existence of optimizers. While for the primal problem this follows from a standard compactness argument, this question is more delicate for the dual problem, even when dealing with probability densities on discrete spaces instead of density matrices, and one has to relax the problem to obtain maximizers (see \cite[Sections 6--7]{GLM19}).

Another interesting direction would be to extend the duality result from matrix algebras to infinite-dimensional systems. While a definition of the metric $\W$ for quantum Markov semigroups on semi-finite von Neumann algebras is available \cite{Hor18,Wir18}, the problem of duality seems to be much harder to address. Even for abstract diffusion semigroups, the best known result only shows that the primal distance is the upper length distance associated with the dual distance and leaves the question of equality open \cite[Proposition 10.11]{AES16}.

\subsection*{Acknowledgments}
The author wants to thank Jan Maas for helpful comments. He acknowledges support from the European Research Council (ERC) under the European Union's Horizon 2020 research and innovation programme (grant agreement No 716117) and from the Austrian Science Fund (FWF) through project F65.

\section{Setting and basic definitions}

In this section we introduce basic facts and definitions about quantum Markov semigroups that will be used later on. In particular, we review the definition of the noncommutative transport distance from \cite{CM17a} and its entropic regularization introduced in \cite{BL21}. Our notation mostly follows \cite{CM17a,CM20}.

Let $M_n(\IC)$ denote the complex $n\times n$ matrices and let $\A$ be a unital $\ast$-subalgebra of $M_n(\IC)$. Let $\A_h$ denote the self-adjoint part of $\A$. We write $\tau$ for the normalized trace on $M_n(\IC)$ and $\H_\A$ for the Hilbert space formed by equipping $\A$ with the GNS inner product
\begin{equation*}
\langle\cdot,\cdot\rangle_{\H_A}\colon \A\times\A\to\IC,\,(A,B)\mapsto \tau(A^\ast B).
\end{equation*}
The adjoint of a linear operator $\mathscr{K}\colon \H_\A\to\H_\A$ is denoted by $\mathscr{K}^\dagger$.

We write $\S(\A)$ for the set of all density matrices on $\A$, that is, all positive elements $\rho\in\A$ with $\tau(\rho)=1$. The subset of invertible density matrices is denoted by $\S_+(\A)$.

A quantum Markov semigroup (QMS) on $\A$ is a family $(P_t)_{t\geq 0}$ of linear operators on $\A$ that satisfy the following conditions:
\begin{itemize}
\item $P_t$ is unital and completely positive for every $t\geq 0$,
\item $P_0=\mathrm{id}_\A$, $P_{s+t}=P_s P_t$ for all $s,t\geq 0$,
\item $t\mapsto P_t$ is continuous.
\end{itemize}
We consider a QMS $(P_t)$ on $\A$ which extends to a QMS on $M_n(\IC)$ satisfying the $\sigma$-DBC for some density matrix $\sigma\in \S_+(A)$, that is,
\begin{equation*}
\tau((P_t A)^\ast B \sigma)=\tau(A^\ast (P_t B)\sigma)
\end{equation*}
for $A,B\in\A$. Let $\L$ denote the generator of $(P_t)$, that is, the linear operator on $\A$ given by
\begin{equation*}
\L(A)=\lim_{t\searrow 0}\frac{P_t A-A}{t}.
\end{equation*}
We further assume that $(P_t)$ is \emph{ergodic} (or primitive), that is, the kernel of $\L$ is one-dimensional.

By Alicki's theorem \cite[Theorem 3]{Ali76}, \cite[Theorem 3.1]{CM17a} there exists a finite set $\J$, real numbers $\omega_j$ for $j\in \J$ and $V_j\in M_n(\IC)$ for $j\in \J$ with the following properties:
\begin{itemize}
\item $\tau(V_j^\ast V_k)=\delta_{jk}$ for $j,k\in\J$,
\item $\tau(V_j)=0$ for $j\in\J$,
\item $\{V_j\mid j\in\J\}=\{V_j^\ast\mid j\in\J\}$,
\item $\sigma V_j \sigma^{-1}=e^{-\omega_j}V_j$ for $j\in \J$
\end{itemize}
such that
\begin{equation*}
\L(A)=\sum_{j\in \J}\left(e^{-\omega_j/2}V_j^\ast[A,V_j]-e^{\omega_j/2}[A,V_j]V_j^\ast\right)
\end{equation*}
for $A\in \A$.

The numbers $\omega_j$ are called Bohr frequencies of $\L$ and are uniquely determined by $(P_t)$.

The matrices $V_j$ are not uniquely determined by $(P_t)$ and $\sigma$, but in the following we will fix a set $\{V_j\mid j\in\J\}$ that satisfies the preceding conditions.

Let
\begin{equation*}
\H_{\A,\J}=\bigoplus_{j\in \J}\H_\A^{(j)},
\end{equation*}
where $\H_\A^{(j)}$ is a copy of $\H_\A$ for $j\in\J$. We write $\partial_j$ for $[V_j,\cdot\,]$ and
\begin{equation*}
\nabla\colon \H_\A\to\H_{\A,\J},\,\nabla(A)=(\partial_j(A))_{j\in\J}.
\end{equation*}
We write $\div$ for the adjoint of $-\nabla$, that is,
\begin{equation*}
\div=-\sum_{j\in\J}\partial_j^\dagger.
\end{equation*}

For $X\in\A_+$ and $\alpha\in \IR$ define
\begin{equation*}
[X]_\alpha\colon \H_\A\to \H_\A,\,[X]_\alpha(A)=\int_0^1 e^{\alpha(s-1/2)}X^s A X^{1-s}\,ds.
\end{equation*}
Given $\vec\alpha=(\alpha_j)_{j\in\J}$, we define
\begin{equation*}
[X]_{\vec\alpha}\colon \H_{\A,\J}\to\H_{\A,\J},\,(A_j)_{j\in\J}\mapsto ([X]_{\alpha_j}A_j)_{j\in\J}.
\end{equation*}

Now let $\vec\omega=(\omega_j)_{j\in\J}$ with the Bohr frequencies $\omega_j$ of $\L$. For $\epsilon\geq 0$ we write $\CE_\epsilon(\rho_0,\rho_1)$ for the set of all pairs $(\rho,\bV)$ such that $\rho\in H^1([0,1];\S_+(\A))$ with $\rho(0)=\rho_0$, $\rho(1)=\rho_1$, $\bV\in L^2([0,1];\H_{\A,\J})$ and
\begin{equation*}
\dot \rho(t) +\div\bV(t)=\epsilon \L^\dagger\rho(t)
\end{equation*}
for a.e. $t\in [0,1]$.

We define a metric $\W_\epsilon$ on $\S_+(\A)$ by
\begin{equation*}
\W_\epsilon^2(\rho_0,\rho_1)=\inf_{(\rho,\bV)\in \CE_\epsilon(\rho_0,\rho_1)}\int_0^1 \langle \bV(t),[\rho(t)]_{\vec\omega}^{-1}\bV(t)\rangle\,dt.
\end{equation*}
A standard mollification argument shows that the infimum can equivalently be taken over $(\rho,\bV)\in \CE_\epsilon(\rho_0,\rho_1)$ with $\rho\in C^\infty([0,1];\S_+(\A))$.

For $\epsilon=0$, this is the noncommutative transport distance $\W$ introduced in \cite{CM17a} (as distance function associated with a Riemannian metric on $\S(\A)_+$), and for $\epsilon>0$, this is the entropic regularization of $\W$ introduced in \cite{BL21}.

By a substitution one can reformulate the minimization problem for $\W_\epsilon$ in such a way that the constraint becomes independent from $\epsilon$. For that purpose define the relative entropy of $\rho\in \S_+(\A)$ with respect to $\sigma$ by
\begin{equation*}
D(\rho\|\sigma)=\tau(\rho(\log\rho-\log \sigma))
\end{equation*}
and the Fisher information of $\rho\in \S_+(\A)$ by
\begin{equation*}
\I(\rho)=\langle [\rho]_{\vec\omega}\nabla (\log \rho-\log\sigma),\nabla(\log\rho-\log\sigma)\rangle_{\H_{\A,\J}}.
\end{equation*}

According to \cite[Theorem 1]{BL21}, one has
\begin{align*}
\W_\epsilon^2(\rho_0,\rho_1)&=\inf_{(\rho,\bW)\in \CE_0(\rho_0,\rho_1)}\int_0^1 (\langle \bW(t),[\rho(t)]_{\vec\omega}^{-1}\bW(t)\rangle+\epsilon^2\I(\rho(t)))\,dt\\
&\quad+2\epsilon(D(\rho_1\|\sigma)-D(\rho_0\|\sigma)).
\end{align*}

The metric $\W$ is intimately connected to the relative entropy and therefore well-suited to study its decay properties along the quantum Markov semigroup. For other applications, variants of the metric $\W$ have also proven useful (e.g. \cite{CGT18,CGGT20}), for which the operator $[\rho]_{\vec\omega}$ is replaced. A systematic framework of these metrics has been developed in \cite{Wir18,CM20}. It can be conveniently phrased in terms of so-called operator connections.

Let $H$ be an infinite-dimensional Hilbert space. A map $\Lambda\colon B(H)_+\times B(H)_+\to B(H)_+$ is called an \emph{operator connection} \cite{KA80} if
\begin{itemize}
\item $A\leq C$ and $B\leq D$ imply $\Lambda(A,B)\leq \Lambda(C,D)$,
\item $C\Lambda(A,B)C\leq \Lambda(CAC,CBC)$,
\item $A_n\searrow A$, $B_n\searrow B$ imply $\Lambda(A_n,B_n)\searrow \Lambda(A,B)$.
\end{itemize}
For example, for every $\alpha\in \IR$ the map
\begin{equation*}
\Lambda_\alpha\colon (A,B)\mapsto \int_0^1 e^{\alpha(s-1/2)}A^s B^{1-s}\,ds
\end{equation*}
is an operator connection.

It can be shown that every operator connection $\Lambda$ satisfies
\begin{equation*}
U^\ast \Lambda(A,B)U=\Lambda(U^\ast AU,U^\ast BU)
\end{equation*}
for $A,B\in B(H)_+$ and unitary $U\in B(H)$ \cite[Section 2]{KA80}. Embedding $\IC^n$ into $H$, one can view $A,B\in M_n(\IC)$ as bounded linear operators on $H$, and the unitary invariance of $\Lambda$ ensures that $\Lambda(A,B)$ does not depend on the embedding of $\IC^n$ into $H$.

For $X\in \A$ define
\begin{align*}
L(X)&\colon \H_{\A}\to \H_{\A},\,A\mapsto XA\\
R(X)&\colon \H_{\A}\to \H_{\A},\,A\mapsto AX.
\end{align*}
With this notation we can write
\begin{equation*}
[\rho]_\Lambda=\Lambda(L(\rho),R(\rho)).
\end{equation*}

Since $L(X)$ and $R(X)$ commute, we have
\begin{equation}\label{eq:spec_rep}
\Lambda(L(X),R(X))A=\sum_{k,l=1}^n \Lambda(\lambda_k,\lambda_l)E_k A E_l\tag{$\ast$}
\end{equation}
for $X\in \A_+$ and $A\in \H_\A$, where $(\lambda_k)$ are the eigenvalues of $X$ and $E_k$ the corresponding spectral projections.

More generally let $\vec\Lambda=(\Lambda_j)_{j\in \J}$ be a family of operator connections and define
\begin{align*}
[\rho]_{\Lambda_j}&=\Lambda_j(L(\rho),R(\rho)),\\
[\rho]_{\vec \Lambda}&=\bigoplus_{j\in \J}[\rho]_{\Lambda_j}.
\end{align*}
Then one can define a distance $\W_{\vec \Lambda}$ by
\begin{align*}
\W_{\vec\Lambda,\epsilon}(\rho_0, \rho_0)^2=\inf_{(\rho,\bV)\in\CE_\epsilon(\rho_0,\rho_1)}\int_0^1 \langle [\rho(t)]_{\vec\Lambda}^{-1}\bV(t),\bV(t)\rangle_{\H_{\A,\J}}\,dt.
\end{align*}
.

If $\Lambda_j=\Lambda_{\omega_j}$ as above, then we retain the original metric $\W_\epsilon$, while for $\Lambda_j(A,B)=\frac 1 2 (A+B)$ (and $\epsilon=0$) one obtains the distance studied in \cite{CGT18,CGGT20}.

%
%

Later we will make the additional assumption that $\Lambda_{j^\ast}(A,B)=\Lambda_j(B,A)$. It follows from the representation theorem of operator means \cite{KA80} that the class of metrics $\W_{\vec\Lambda,0}$ with $\vec\Lambda$ subject to this symmetry condition is exactly the class of metrics satisfying Assumptions 7.2 and 9.5 in \cite{CM20}.

For technical reasons, it can be useful to allow for curves of density matrices that are not necessarily invertible. For this purpose, we make the following convention: If $\mathcal{K}\colon \H_{\A,\J}\to\H_{\A,\J}$ is a positive operator and $\bV\in\H_{\A,\J}$, we define
\begin{equation*}
\langle \bV,\mathcal{K}^{-1}\bV\rangle_{\H_{\A,\J}}=\begin{cases}\langle\mathcal{K}\bW,\bW\rangle_{\H_{\A,\J}}&\text{if }\bV\in (\ker\mathcal{K})^\perp,\mathcal{K}\bW=\bV,\\\infty&\text{otherwise}.\end{cases}
\end{equation*}
Since $(\ker \mathcal{K})^\perp=\operatorname{ran}\mathcal{K}$ and $\mathcal{K}$ is injective on $(\ker\mathcal{K})^\perp$, the element $\bW$ in this definition exists and is unique. Moreover, this convention is clearly consistent with the usual definition if $\mathcal{K}$ is invertible.

Alternatively, as a direct consequence of the spectral theorem, $\langle \bV,\mathcal K^{-1}\bV\rangle_{\H_{\A,\J}}$ can be expressed as
\begin{equation*}
\langle \bV,\mathcal{K}^{-1}\bV\rangle_{\H_{\A,\J}}=\sum_{k=1}^m \frac 1{\lambda_k}\abs{\langle\bV,\bW_k\rangle_{\H_{\A,\J}}}^2,
\end{equation*}
where $\lambda_1,\dots,\lambda_m$ are the eigenvalues of $\mathcal{K}$ and $\W_1,\dots,W_m$ an orthonormal basis of corresponding eigenvectors.

\begin{lemma}\label{lem:monotone_conv}
If $\mathcal K_n\colon \H_{\A,\J}\to\H_{\A,\J}$, $n\in \mathbb N$, are positive invertible operators that converge monotonically decreasing to $\mathcal{K}$, then
\begin{equation*}
\langle \bV,\mathcal K_n^{-1}\bV\rangle_{\H_{\A,\J}}\nearrow \langle \langle \bV,\mathcal K^{-1}\bV\rangle_{\H_{\A,\J}}
\end{equation*}
for all $\bV\in\H_{\A,\J}$.
\end{lemma}
\begin{proof}
From the spectral expression it is easy to see that
\begin{equation*}
\langle\bV,\mathcal{K}_n^{-1}\bV\rangle_{\H_{\A,\J}}=\sup_{\delta>0}\langle \bV,(\mathcal{K}_n+\delta)^{-1}\bV\rangle_{\H_{\A,\J}}
\end{equation*}
and the same for $\mathcal{K}_n$ replaced by $\mathcal{K}$. Moreover, since $\mathcal{K}_n\searrow \mathcal{K}$, we have $(\mathcal{K}_n+\delta)^{-1}\nearrow (\mathcal{K}+\delta)^{-1}$. Thus
\begin{align*}
\langle \bV,\mathcal{K}^{-1}\bV\rangle_{\H_{\A,\J}}&=\sup_{\delta>0}\langle \bV,(\mathcal{K}+\delta)^{-1}\bV\rangle_{\H_{\A,\J}}\\
&=\sup_{\delta>0}\sup_{n\in\mathbb N}\langle\bV,(\mathcal{K}_n+\delta)^{-1}\bV\rangle_{\H_{\A,\J}}\\
&=\sup_{n\in\mathbb N}\sup_{\delta>0}\langle\bV,(\mathcal{K}_n+\delta)^{-1}\bV\rangle_{\H_{\A,\J}}\\
&=\sup_{n\in\mathbb N}\langle \bV,\mathcal{K}_n^{-1}\bV\rangle_{\H_{\A,\J}}.
\end{align*}
Since $(\langle\bV,\mathcal{K}_n^{-1}\bV\rangle_{\H_{\A,\J}})$ is monotonically increasing, this settles the claim.
\end{proof}

Write $\CE_{\epsilon}^\prime(\rho_0,\rho_1)$ for the set of all pairs $(\rho,\bV)$ such that $\rho\in H^1([0,1];\S(\A))$ with $\rho(0)=\rho_0$, $\rho(1)=\rho_1$, $\bV\in L^2([0,1];\H_{\A,\J})$ and
\begin{equation*}
\dot \rho(t)+\div\bV(t)=\epsilon\L^\dagger\rho(t)
\end{equation*}
for a.e. $t\in[0,1]$. The only difference to the definition of $\CE_\epsilon(\rho_0,\rho_1)$ is that $\rho(t)$ is not assumed to be invertible.

\begin{proposition}\label{prop:relaxed_paths_CE}
For $\rho_0,\rho_1\in \S_+(\A)$ we have
\begin{equation*}
\W_{\vec\Lambda,\epsilon}^2(\rho_0,\rho_1)=\inf_{(\rho,\bV)\in \CE_\epsilon^\prime(\rho_0,\rho_1)}\int_0^1 \langle \bV(t),[\rho(t)]_{\vec\Lambda}^{-1}\bV(t)\rangle\,dt.
\end{equation*}
\end{proposition}
\begin{proof}
It suffices to show that every curve $(\rho,\bV)\in \CE_\epsilon^\prime(\rho_0,\rho_1)$ can be approximated by curves in $\CE_\epsilon(\rho_0,\rho_1)$ such that the action integrals converge.

For that purpose let
\begin{equation*}
\rho^\delta\colon [0,1]\to \S_+(\A),\,t\mapsto\begin{cases}(1-t)\rho_0+t1_\A&\text{if }t\in[0,\delta],\\
(1-\delta)\rho((1-2\delta)^{-1}(t-\delta))+\delta 1_\A&\text{if }t\in(\delta,1-\delta),\\
t\rho_1+(1-t)1_\A&\text{if }t\in[1-\delta,1].
\end{cases}
\end{equation*}
Since $(P_t)$ is assumed to be ergodic, for $t\in[0,\delta]$ there exists $\bV^\delta(t)=\nabla U^\delta(t)$ such that
\begin{equation*}
\dot \rho^\delta(t)+\div\bV^\delta(t)=1-\rho_0+\div \bV^\delta(t)=\epsilon (1-t)\L^\dagger\rho_0=\epsilon \L^\dagger\rho^\delta(t)
\end{equation*}
and
\begin{equation*}
\norm{\bV^\delta(t)}_{\H_{\A,\J}}\leq C\norm{\epsilon(1-t)\L^\dagger\rho_0-1+\rho_0}_{\H_\A}
\end{equation*}
with a constant $C>0$ (depending only on the spectral gap of $\L$). In particular, it is bounded independent of $\delta$.

Moreover, if $\lambda$ is the smallest eigenvalue of $\rho_0$, which is strictly positive by assumption, then $\rho^\delta(t)\geq ((1-t)\lambda+t)1_\A\geq \lambda 1_\A$.

Thus
\begin{align*}
\int_0^\delta \langle \bV^\delta(t),[\rho^\delta(t)]_{\vec\Lambda}^{-1}\bV^\delta(t)\rangle_{\H_{\A,\J}}\,dt&\leq \int_0^\delta \langle \bV^\delta(t),[\lambda 1_\A]_{\vec\Lambda}^{-1}\bV^\delta(t)\rangle_{\H_{\A,\J}}\,dt\\
&\leq \norm{[\lambda 1_\A]_{\vec \Lambda}^{-1}}\int_0^\delta\norm{\bV^\delta(t)}_{\H_{\A,\J}}^2\,dt\\
&\to 0
\end{align*}
as $\delta\to 0$. Similarly one can show
\begin{equation*}
\lim_{\delta\to 0}\int_{1-\delta}^1 \langle \bV^\delta(t),[\rho^\delta(t)]_{\vec\Lambda}^{-1}\bV^\delta(t)\rangle_{\H_{\A,\J}}\,dt=0.
\end{equation*}
By the same argument as above, for a.e. $t\in (\delta,1-\delta)$ there exists a unique gradient $\bW^\delta(t)$ such that
\begin{equation*}
\div \bW^\delta(t)=-\frac{2\delta\epsilon}{1-2\delta}\L^\dagger \rho((1-2\delta)^{-1}(t-\delta))
\end{equation*}
and
\begin{equation*}
\norm{\bW^\delta(t)}_{\H_{\A,\J}}\leq \frac{2\delta\epsilon}{1-2\delta}\norm{\L^\dagger\rho((1-2\delta)^{-1}(t-\delta))}_{\H_{\A,\J}}.
\end{equation*}
Since $\rho\in H^1([0,1];\S(\A))\subset C([0,1];\S(\A))$, the norm on the right side is bounded independent of $\delta$, so that
\begin{equation*}
\norm{\bW^\delta(t)}_{\H_{\A,\J}}\leq \tilde C\delta
\end{equation*}
with a constant $\tilde C>0$ independent of $\delta$. As $\rho^\delta(t)\geq \delta 1_\A$ for $t\in (\delta,1-\delta)$, this implies
\begin{align*}
\int_\delta^{1-\delta}\langle\bW^\delta(t), [\rho^\delta(t)]_{\vec\Lambda}^{-1}\bW^\delta(t)\rangle_{\H_{\A,\J}}\,dt&\leq \frac 1 \delta\int_\delta^{1-\delta}\langle \bW^\delta(t),[1_\A]_{\vec\Lambda}^{-1}\bW^\delta(t)\rangle_{\H_{\A,\J}}\,dt\\
&\leq \tilde C \norm{[1_\A]_{\vec\Lambda}^{-1}}\delta\\
&\to 0
\end{align*}
as $\delta\to 0$.

With 
\begin{equation*}
\bV^\delta(t)=\frac 1{1-2\delta}\bV((1-2\delta)^{-1}(t-\delta))+\bW^\delta(t)
\end{equation*}
we have
\begin{equation*}
\dot \rho^\delta(t)+\div \bV^\delta(t)=\epsilon \L\rho^\delta(t).
\end{equation*}
Furthermore,
\begin{align*}
&\quad\int_\delta^{1-\delta}\left\langle\bV\left(\frac{t-\delta}{1-2\delta}\right),[\rho^\delta(t)]_{\vec\Lambda}^{-1}\bV\left(\frac{t-\delta}{1-2\delta}\right)\right\rangle_{\H_{\A,\J}}\,dt\\
&=\frac 1{1-\delta}\int_\delta^{1-\delta}\left\langle\bV\left(\frac{t-\delta}{1-2\delta}\right),\left[\rho\left(\frac{t-\delta}{1-2\delta}\right)+\frac{\delta}{1-\delta}\right]_{\vec\Lambda}^{-1}\bV\left(\frac{t-\delta}{1-2\delta}\right)\right\rangle_{\H_{\A,\J}}\,dt\\
&=\frac{1-2\delta}{1-\delta}\int_0^1 \langle \bV(s),\left[\rho(s)+\frac{\delta}{1-\delta}\right]_{\vec\Lambda}^{-1}\bV(s)\rangle_{\H_{\A,\J}}\,ds,
\end{align*}
where we used the substitution $s=(1-2\delta)^{-1}(t-\delta)$ in the last step.

By Lemma \ref{lem:monotone_conv} and the monotone convergence theorem we obtain
\begin{align*}
&\quad\;\lim_{\delta\to 0}\int_\delta^{1-\delta}\left\langle\bV\left(\frac{t-\delta}{1-2\delta}\right),[\rho^\delta(t)]_{\vec\Lambda}^{-1}\bV\left(\frac{t-\delta}{1-2\delta}\right)\right\rangle_{\H_{\A,\J}}\,dt\\
&=\int_0^1 \langle \bV(s),[\rho(s)]_{\vec\Lambda}^{-1}\bV(s)\rangle_{\H_{\A,\J}}\,ds.
\end{align*}

Together with the convergence result for $\bW^\delta$ from above, this implies
\begin{equation*}
\int_\delta^{1-\delta}\langle \bV^\delta(t),[\rho^\delta(t)]_{\vec \Lambda}^{-1}\bV^\delta(t)\rangle_{\H_{\A,\J}}\,dt\to \int_0^1\langle \bV(t),[\rho(t)]_{\vec\Lambda}^{-1}\bV(t)\rangle_{\H_{\A,\J}}\,dt.
\end{equation*}
Altogether we have shown
\begin{equation*}
\lim_{\delta\to 0}\int_0^{1}\langle \bV^\delta(t),[\rho^\delta(t)]_{\vec \Lambda}^{-1}\bV^\delta(t)\rangle_{\H_{\A,\J}}\,dt= \int_0^1\langle \bV(t),[\rho(t)]_{\vec\Lambda}^{-1}\bV(t)\rangle_{\H_{\A,\J}}\,dt.\qedhere
\end{equation*}
\end{proof}

\section{Real subspaces}

Since the proof of the main result relies on convex analysis methods for real Banach spaces, we need to identify suitable real subspaces for our purposes. For $\A$ this is simply $\A_h$, but for $\H_{\A,\J}$ this is less obvious and will be done in the following.

For $j\in\J$ denote by $j^\ast$ the unique index in $\J$ such that $V_j^\ast=V_{j^\ast}$. Let $\tilde \H_ {\A}^{(j)}$ be the linear span of $\{X\partial_j A\mid A,X\in\A\}$, and define a linear map $J\colon \tilde \H_{\A}^{(j)}\to \tilde \H_{\A}^{(j^\ast)}$ by 
\begin{align*}
J(X\partial_j A)=\partial_{j^\ast} (A^\ast)X^\ast.
\end{align*}
By the product rule, $(\partial_j A)X$ also belongs to $\tilde \H_{\A,\J}^{(j)}$ and $J((\partial_j A)X)=X^\ast \partial_{j^\ast}(A^\ast)$.

\begin{lemma}
The map $J$ is anti-unitary.
\end{lemma}
\begin{proof}
For $A,B,X,Y\in\A$ we have
\begin{align*}
\langle J(X\partial_j A),J(Y\partial_j B)\rangle_{\H_\A}&=\tau(X(AV_j-V_A)(V_j^\ast B^\ast-B^\ast V_j^\ast) Y^\ast)\\
&=\tau((B^\ast V_j^\ast-V_j^\ast B^\ast)Y^\ast X(V_j A-AV_j))\\
&=\langle Y\partial_j B,X\partial_j A\rangle.\qedhere
\end{align*}
\end{proof}

Let 
\begin{equation*}
\Hh_{\A,\J}=\{\bV\in \bigoplus_{j\in\J}\tilde\H_{\A}^{(j)}\mid J(\bV_j)=\bV_{j^\ast}\}.
\end{equation*}
By the previous lemma, $\Hh_{\J,\A}$ is a real Hilbert space.

\begin{lemma}\label{lem:grad_real}
Let $(\Lambda_j)_{j\in\J}$ be a family of operator connections such that $\Lambda_{j^\ast}(A,B)=\Lambda_j(B,A)$ for all $j\in\J$. If $A\in \A_h$ and $\rho\in \S(\A)$, then $\nabla A,[\rho]_{\vec\Lambda}\nabla A\in \Hh_{\A,\J}$.
\end{lemma}
\begin{proof}
For $\nabla A$ the statement follows directly from the definitions. For $[\rho]_{\vec\Lambda} \nabla A$ first note that 
\begin{equation*}
J\Lambda(L(\rho),R(\rho))=\Lambda(R(\rho),L(\rho))J
\end{equation*}
as a consequence of the spectral representation (\ref{eq:spec_rep}).

Thus
\begin{align*}
J([\rho]_{\Lambda_j}\partial_j A)&=J\Lambda_j(L(\rho),R(\rho))\partial_j A\\
&=\Lambda_j(R(\rho),L(\rho))J\partial_j A\\
&=\Lambda_{j^\ast}(L(\rho),R(\rho))\partial_{j^\ast}A.\qedhere
\end{align*}
\end{proof}

\section{Duality}

In this section we prove the duality theorem announced in the introduction. Our strategy follows the same lines as the proof in the commutative case in \cite{EMW19}. It crucially relies on the Rockefellar--Fenchel duality theorem quoted below. Throughout this section we fix a quantum Markov semigroup with generator $\L$ satisfying the $\sigma$-DBC for some $\sigma\in \S_+(\A)$ and a family $(\Lambda_j)_{j\in\J}$ of operator connections such that $\Lambda_{j^\ast}(A,B)=\Lambda_j(B,A)$ for all $j\in \J$.

We need the following definition for the constraint of the dual problem. Here and in the following we write
\begin{equation*}
\langle \bV,\bW\rangle_\rho=\langle \bV,[\rho]_{\vec\Lambda}\bW\rangle_{\H_{\A,\J}}
\end{equation*}
for $\bV,\bW\in \H_{\A,\J}$ and $\rho\in \A_+$.

\begin{definition}\label{def:HJB}
A function $A\in H^1((0,T);\A_h)$ is said to be a \emph{Hamilton--Jacobi--Bellmann subsolution} if for a.e. $t\in (0,T)$ we have
\begin{align*}
\tau((\dot A(t)+\epsilon \L A(t)) \rho)+\frac 1 2 \norm{\nabla A(t)}_{\rho}^2\leq 0\qquad\text{for all }\rho\in \S(A).
\end{align*}
The set of all Hamilton--Jacobi--Bellmann subsolutions is denoted by $\HJB_{\vec\Lambda,\epsilon}$.
\end{definition}

Our proof will establish equality between the primal and dual problem, but before we begin, let us show that one inequality is actually quite easy to obtain.

\begin{proposition}
For all $\rho_0,\rho_1\in\S_+(\A)$ we have
\begin{equation*}
\sup_{A\in\HJB_{\vec\Lambda,\epsilon}}\tau(A(1)\rho_1-A(1)\rho_0)\leq\frac 1 2 \inf_{(\rho,\bV)\in\CE_\epsilon(\rho_0,\rho_1)}\int_0^1 \langle \bV(t),[\rho(t)]_{\vec\Lambda}^{-1}\bV(t)\rangle_{\H_{\A,\J}}\,dt.
\end{equation*}
\end{proposition}
\begin{proof}
For $A\in \HJB_{\vec\Lambda,\epsilon}$ and $(\rho,\bV)\in \CE_\epsilon(\rho_0,\rho_1)$ we have
\begin{align*}
\tau(A(1)\rho_1-A(0)\rho_0)&=\int_0^1 \tau(\dot A(t) \rho(t)+A(t)\dot\rho(t))\,dt\\
&\leq-\int_0^1 (\epsilon\tau((\L A(t))\rho(t))+\frac 1 2\norm{\nabla A(t)}_{\rho(t)}^2)\,dt\\
&\quad+\int_0^1(\langle \nabla A(t),\bV(t)\rangle_{\H_{\A,\J}}+\epsilon\tau(A(t)\L^\dagger\rho(t)))\,dt\\
&= \int_0^1\langle  [\rho(t)]_{\vec\omega}^{1/2}\nabla A(t),[\rho(t)]_{\vec\Lambda}^{-1/2}\bV(t)\rangle_{\H_{\A,\J}}\,dt\\
&\quad-\frac 1 2 \int_0^1\langle [\rho(t)]_{\vec\Lambda}^{1/2}\nabla A(t),[\rho(t)]_{\vec \Lambda}^{1/2}\nabla A(t)\rangle_{\H_{\A,\J}}\,dt\\
&\leq\frac 1 2\int_0^1\langle\bV(t),[\rho(t)]_{\vec\Lambda}^{-1}\bV(t)\rangle_{\H_{\A,\J}}\,dt,
\end{align*}
where we used $A\in \HJB_{\vec\Lambda,\epsilon}$ and $(\rho,\bV)\in \CE_\epsilon(\rho_0,\rho_1)$ for the first inequality and Young's inequality for the second inequality.
\end{proof}

To prove actual equality, our crucial tool is the Rockefellar--Fenchel duality theorem (see e.g. \cite[Theorem 1.9]{Vil03}, which we quote here for the convenience of the reader. Recall that if $E$ is a (real) normed space, the Legendre--Fenchel transform $F^\ast$ of a proper convex function $F\colon E\to \IR\cup\{\infty\}$ is defined by
\begin{equation*}
F^\ast\colon E^\ast\to\mathbb{R}\cup\{\infty\},\,F^\ast(x^\ast)=\sup_{x\in E}(\langle x^\ast,x\rangle-F(x)).
\end{equation*}

\begin{theorem}\label{thm:Rockefellar-Fenchel}
Let $E$ be a real normed space and $F,G\colon E\to \IR\cup\{\infty\}$ proper convex functions with Legendre--Fenchel transforms $F^\ast,G^\ast$. If there exists $z_0\in E$ such $G$ is continuous at $z_0$ and $F(z_0),G(z_0)<\infty$, then 
\begin{equation*}
\sup_{z\in E}(-F(z)-G(z))=\min_{z^\ast\in E^\ast}(F^\ast(z^\ast)+G^\ast(-z^\ast)).
\end{equation*}
\end{theorem}

Before we state the main result, we still need the following useful inequality.

\begin{lemma}\label{lem:deriv_conv_fct}
For any operator connection $\Lambda$ the map
\begin{equation*}
f_{\Lambda}\colon \A_{++}\to B(\H_A),\,A\mapsto [A]_\Lambda
\end{equation*}
is smooth and its Fréchet derivative satisfies
\begin{align*}
d f_\Lambda(B) A\geq f_\Lambda(A)
\end{align*}
for $A,B\in \A_{++}$ with equality if $A=B$.
\end{lemma}
\begin{proof}
Smoothness of $f_\Lambda$ is a consequence of the representation theorem of operator connections [Theorem 3.4]\cite{KA80}. For the claim about the Fréchet derivative first note that $f_\Lambda$ is concave \cite[Theorem 3.5]{KA80}. Therefore $d^2 f_\Lambda(X)[Y,Y]\leq 0$ for all $X\in \A_{++}$ and $Y\in \A_h$ by \cite[Proposition 2.2]{Han97}.

The fundamental theorem of calculus implies
\begin{align*}
(d f_\Lambda(A)-d f_\Lambda(B))(A-B)&=\int_0^1 d^2 f_\Lambda(tA+(1-t)B)[A-B,A-B]\,dt\\
&\leq 0.
\end{align*}
Since $f_\Lambda$ is $1$-homogeneous, its derivative is $0$-homogeneous. Thus, if we replace $B$ by $\epsilon B$ and let $\epsilon\searrow$, we obtain
\begin{equation*}
d f_\Lambda(A)A\leq d f_\Lambda(B)A.
\end{equation*}
Moreover, the $1$-homogeneity of $f_\Lambda$ implies $d f_\Lambda(A)A=f_\Lambda(A)$, which settles the claim.
\end{proof}

\begin{theorem}[Duality formula]\label{thm:duality}
For $\rho_0,\rho_1\in \S_+(A)$ we have
\begin{align*}
\frac1  2\W_{\vec\Lambda,\epsilon}(\rho_0,\rho_1)^2&=\sup\{\tau(A(1)\rho_1)-\tau(A(0),\rho_0) : A\in \HJB_{\vec\Lambda,\epsilon}\}\\
&=\sup\{\tau(A(1)\rho_1)-\tau(A(0),\rho_0) : A\in \HJB_{\vec\Lambda,\epsilon}\cap C^\infty([0,1];\A)\}.
\end{align*}
\end{theorem}
\begin{proof}
The second inequality follows easily by mollifying. We will show the duality formula for Hamilton--Jacobi subsolutions in $H^1$. For this purpose we use the Rockefellar--Fenchel duality formula from Theorem \ref{thm:Rockefellar-Fenchel}.

Let $E$ be the real Banach space 
\begin{equation*}
H^1([0,1];\Hh_\A)\times L^2([0,1];\Hh_{\A,\J}).
\end {equation*}
By the theory of linear ordinary differential equations, the map
\begin{equation*}
H^1([0,1];\Hh_\A)\to\Hh_\A\times L^2([0,1];\Hh_\A),\,A\mapsto (A(0),\dot A+\epsilon \L A)
\end{equation*}
is a linear isomorphism.

Thus the dual space $E^\ast$ can be isomorphically identified with
\begin{align*}
\Hh_\A\times L^2([0,1];\Hh_\A)\times L^2([0,1];\Hh_{\A,\J})
\end{align*}
via the dual pairing
\begin{align*}
\langle (A,\bV),(B,C,\bW)\rangle&=\tau(A(0)B)+\int_0^1 \tau((\dot A(t)+\epsilon \L A(t))C(t))\,dt\\
&\quad+\int_0^1 \langle\bV(t),\bW(t)\rangle_{\H_{\A,\J}}\,dt.
\end{align*}
Define functionals $F,G\colon E\lra \IR\cup\{\infty\}$ by
\begin{align*}
F(A,\bV)&=\begin{cases}-\tau(A(1)\rho_1)+\tau(A(0)\rho_0)&\text{if }\bV=\nabla A,\\
\infty&\text{otherwise},
\end{cases}\\
G(A,\bV)&=\begin{cases}
0&\text{if }(A,\bV)\in \D,\\
\infty&\text{otherwise}.
\end{cases}
\end{align*}
Here 
\begin{equation*}
\D=\{(A,\bV): \tau((\dot A(t)+\epsilon \L A(t))\rho)+\frac 1 2\norm{\bV(t)}_{\rho}^2\leq 0\text{ for all }t\in [0,1],\rho\in \S(\A)\}.
\end{equation*}
It is easy to see that $F$ and $G$ are convex. Moreover, for $A_0(t)=-t 1_\A$ and $\bV_0=0$ we have $\bV_0=\nabla A_0$, hence $F(A_0,\bV_0)=0$, and
\begin{equation*}
\tau((\dot A_0(t)+\epsilon \L A_0(t)) \rho)+\frac 1 2\norm{\bV_0(t)}_{\rho}^2=-1
\end{equation*}
for all $t\in[0,1],\;\rho\in\S(A)$, hence $G(A_0,\bV_0)=0$. Furthermore, $G$ is clearly continuous at $(A_0,\bV_0)$.

Moreover,
\begin{equation*}
\sup_{(A,\bV)\in E}(-F(A,\bV)-G(A,\bV))=\sup_{A\in\HJB_{\vec\Lambda,\epsilon}(\rho_0,\rho_1)}(\tau(A(1)\rho_1)-\tau(A(0)\rho_0)).
\end{equation*}

Let us calculate the Legendre transforms of $F$ and $G$, keeping in mind the identification of $E^\ast$. For $F$ we obtain
\begin{align*}
F^\ast(B,C,\bW)&=\sup_{(A,\bV)\in E}\bigg\{\langle (A,\bV),(B,C,\bW)\rangle-F(A,\bV)\bigg\}\\
&=\sup_{A}\bigg\{\tau(A(0)B)+\int_0^1 \tau((\dot A(t)+\epsilon \L A(t))\,C(t))\,dt\\
&\quad+\int_0^1 \langle\nabla A(t),\bW(t)\rangle_{\H_{A,\J}}\,dt+\tau(A(1)\rho_1)-\tau(A(0)\rho_0)\bigg\}.
\end{align*}
Since the last expression is homogeneous in $A$, we have $F^\ast(B,C,\bW)=\infty$ unless
\begin{align*}
-\tau(A(1)\rho_1)+\tau(A(0)(\rho_0-B))&=\int_0^1 \tau((\dot A(t)+\epsilon \L A(t))\,C(t))\,dt\\
&\quad+\int_0^1 \langle \nabla A(t),\bW(t)\rangle_{\H_{\A,\J}}\,dt
\end{align*}
for $A\in H^1([0,1];\Hh_\A)$.

This implies $C_0=-(\rho_0-B)$ and $C_1=-\rho_1$ and 
\begin{equation*}
\dot C(t)+\div \bW(t)=\epsilon \L^\dagger C(t).
\end{equation*}
Thus
\begin{equation*}
F^\ast(B,C,\bW)=\begin{cases}0&\text{if }(-C,-\bW)\in \CE_\epsilon^{\prime\prime}(\rho_0-B,\rho_1),\\
\infty&\text{otherwise}.
\end{cases}
\end{equation*}
Here $\CE_\epsilon^{\prime\prime}(\rho_0-B,\rho_1)$ denotes the set of all pairs $(X,\mathbf{U})\in H^1((0,1);\Hh_\A)\times L^2((0,1);\Hh_{\A,\J})$ satisfying $X(0)=\rho_0-B$, $X(1)=\rho_1$ and
\begin{equation*}
\dot X(t)+\div \bm{U}(t)=\epsilon \L^\dagger X(t).
\end{equation*}

The difference to the definitions of $\CE$ (or $\CE^\prime$) and $\CE^{\prime\prime}$ is that we do not make any positivity or normalization constraints. Note however that if $(X,\bm U)\in \CE^{\prime\prime}(\rho_0-B,\rho_1)$, then
\begin{equation*}
\frac{d}{dt}\tau(X(t))=\tau(\epsilon \L^\dagger X(t)-\div \bm U(t))=0
\end{equation*}
so that $\tau(X(t))=\tau(\rho_1)=1$ (and $\tau(B)=0$).

Now let us turn to the Legendre transform of $G$. We have
\begin{align*}
G^\ast(B,C,\bW)&=\sup_{(A,\bV)\in E}\bigg\lbrace \langle (A,\bV),(B,C,\bW)\rangle-G(A,\bV)\bigg\rbrace\\
&=\sup_{(A,\bV)\in \D}\bigg\lbrace \tau(A(0)B)+\int_0^1 \tau((\dot A(t)+\epsilon \L A(t))C(t))\,dt\\
&\qquad\qquad+\int_0^1\langle \bV(t),\bW(t)\rangle_{\H_{\A,\J}})\,dt\bigg\rbrace.
\end{align*}
Since $(A,\bV)\in \D$ implies $(A+X,\bV)\in \D$ for all $X\in \A$, we have $G^\ast(B,C,\bV)=\infty$ unless $B=0$. Moreover, it follows from the definition of $\D$ that $G^\ast(0,C,\bW)=\infty$ unless $C(t)\geq 0$ for a.e. $t\in [0,1]$.

For $B=0$ we have
\begin{align*}
G^\ast(0,C,\bW)&=\sup_{(A,\bV)\in\D}\left\lbrace\int_0^1 (\tau((\dot A(t)+\epsilon \L A(t))C(t))+\langle \bV(t),\bW(t)\rangle_{\H_{\A,\J}})\,dt\right\rbrace\\
&\leq \sup_{(A,\bV)\in\D}\bigg\lbrace -\int_0^1 \frac 1 2\norm{\bV(t)}_{C(t)}^2\,dt\\
&\qquad\qquad+\int_0^1\langle [C(t)]_{\vec\Lambda}^{1/2}\bV(t),[C(t)]_{\vec\Lambda}^{-1/2}\bW(t)\rangle_{\H_{\A,\J}}\,dt\bigg\rbrace\\
&\leq \frac 1 2\int_0^1\langle [C(t)]_{\vec\Lambda}^{-1} \bW(t),\bW(t)\rangle_{\H_{\A,\J}}\,dt.
\end{align*}
We will show next that the inequalities are in fact equalities. Let $C^\delta=C+\delta$
%
and $\bV^{\delta}(t)=[C(t)^{\delta}]^{-1}\bW(t)$. Moreover, let $f_j=f_{\Lambda_j}$ with the notation from Lemma \ref{lem:deriv_conv_fct}. Since 
\begin{equation*}
\Hh_\A\to \IR,\,B\mapsto \sum_{j\in\J}\langle (df_{j}(C^{\delta}(t))B)\bV^{\delta}_j(t),\bV^{\delta}_j(t)\rangle_{\H_\A}
\end{equation*}
is a bounded linear map that depends continuously on $t$, there exists a unique continuous map $X^{\delta}\colon [0,1]\to \Hh_\A$ such that
\begin{equation*}
\tau(B X^{\delta}(t))=\sum_{j\in\J}\langle (df_{j}(C^{\delta}(t))B)\bV^{\delta}_j(t),\bV^{\delta}_j(t)\rangle_{\H_\A}
\end{equation*}
for every $B\in\Hh_\A$ and $t\in[0,1]$.

Let 
\begin{equation*}
 A^{\delta}\colon [0,1]\to \A_h,\,A^{\delta}(t)=-\frac1  2\int_0^t X^{\delta}(s)\,ds.
\end{equation*}
We claim that $(A^{\delta},\bV^{\delta})\in \D$. Indeed, 
\begin{align*}
\tau(\dot A^{\delta}(t)\rho)&=-\frac 1 2\sum_{j\in\J}\langle (df_{j}(C^{\delta}(t))\rho)\bV^{\delta}_j(t),\bV^{\delta}_j(t)\rangle_{\H_\A}\\
&\leq -\frac 1 2\sum_{j\in\J}\langle [\rho]_{\Lambda_j}\bV^{\delta}_j(t),\bV^{\delta}_j(t)\rangle_{\H_\A}\\
&=-\frac 1 2\norm{\bV^{\delta}(t)}_{\rho}^2,
\end{align*}
where the inequality follows from Lemma \ref{lem:deriv_conv_fct}. Note that we have equality for $\rho=C^{\delta}(t)$.

In particular, for $\rho=C(t)$ we obtain
\begin{equation*}
\tau(\dot A^{\delta}(t)C(t))+\langle \bV^{\delta}(t),\bW(t)\rangle_{\H_{\A,\J}}\leq \frac 1 2\langle [C(t)]_{\vec\Lambda}^{-1}\bW(t),\bW(t)\rangle_{\H_{\A,\J}}.
\end{equation*}
On the other hand,
\begin{align*}
\tau(\dot A^{\delta}(t)C(t))&=-\frac 1 2\sum_{j\in\J}\langle(d f_j(C^\delta(t))(C^\delta(t)-\delta))\bV_j^\delta(t),\bV_j^\delta(t)\rangle_{\H_A}\\
&\geq-\frac 1 2\langle [C^\delta(t)]_{\vec\Lambda}\bV^\delta(t),\bV^\delta(t)\rangle_{\H_{\A,\J}}+\frac 1 2\langle [\delta]_{\vec \Lambda}\bV^\delta(t),\bV^{\delta}(t)\rangle_{\H_{\A,\J}}\\
&\geq -\frac 1 2 \langle \bV^\delta(t),\bW(t)\rangle_{\H_{\A,\J}},
\end{align*}
where we again used Lemma \ref{lem:deriv_conv_fct} for the first inequality.

Put together, we have
\begin{align*}
\frac 1 2\langle [C^\delta(t)]_{\vec\Lambda}^{-1}\bW(t),\bW(t)\rangle_{\H_{\A,\J}}&\leq \tau(\dot A^{\delta}(t)C(t))+\langle \bV^{\delta}(t),\bW(t)\rangle_{\H_{\A,\J}}\\
&\leq \frac 1 2\langle [C(t)]_{\vec\Lambda}^{-1}\bW(t),\bW(t)\rangle_{\H_{\A,\J}},
\end{align*}
and
\begin{align*}
\int_0^1 (\tau(\dot A^\delta(t)C(t))+\langle \bV^\delta(t),\bW(t)\rangle_{\H_{\A,\J}})\,dt\to \frac 1 2\int_0^1 \langle [C(t)]_{\vec\Lambda}^{-1}\bW(t),\bW(t)\rangle_{\H_{\A,\J}}\,dt
\end{align*}
follows from the monotone convergence theorem.

Hence 
\begin{equation*}
G^\ast(0,C,\bW)=\frac 1 2\int_0^1 \langle [C(t)]_{\vec\Lambda}^{-1}\bW(t),\bW(t)\rangle_{\H_{\A,\J}}\,dt
\end{equation*}
if $C(t)\geq 0$ for a.e. $t\in [0,1]$. Together with the formula for $F^\ast$, we obtain
\begin{align*}
&\quad\inf_{(B,C,\bW)\in E^\ast}(F^\ast(-B,-C,-\bW)+G^\ast(B,C,\bW))\\
&=\inf_{(\rho,\bV)\in \CE_\epsilon^\prime(\rho_0,\rho_1)}\frac 1 2 \int_0^1 \langle [\rho(t)]_{\vec\Lambda}^{-1}\bV(t),\bV(t)\rangle_{\H_{\A,\J}}\,dt\\
&=\frac 1 2 \W_{\vec\Lambda,\epsilon}^2(\rho_0,\rho_1),
\end{align*}
where the last equality follows from Proposition \ref{prop:relaxed_paths_CE}.

An application of the Rockefellar--Fenchel theorem yields the desired conclusion.
\end{proof}

\DeclareFieldFormat[article]{citetitle}{#1}
\DeclareFieldFormat[article]{title}{#1} 

\printbibliography

\end{document}